\documentclass[english]{article}
\usepackage{amsmath}
\usepackage{amsthm}
\usepackage{newtxmath}
\usepackage[LGR,T1]{fontenc}
\usepackage[utf8]{inputenc}
\usepackage{babel}
\usepackage{refstyle}
\usepackage{mathtools}
\usepackage{geometry}
\usepackage{wasysym}
\usepackage[pdfusetitle,
 bookmarks=true,bookmarksnumbered=false,bookmarksopen=false,
 breaklinks=false,pdfborder={0 0 1},backref=false,colorlinks=false]
 {hyperref}

\makeatletter

\RS@ifundefined{subsecref}{
  \newref{subsec}{
        name      = \RSsectxt,
        names     = \RSsecstxt,
        Name      = \RSSectxt,
        Names     = \RSSecstxt,
        refcmd    = {\S\ref{#1}},
        rngtxt    = \RSrngtxt,
        lsttwotxt = \RSlsttwotxt,
        lsttxt    = \RSlsttxt
  }
}{}

\@ifundefined{date}{}{\date{}}
\usepackage{babel}

\usepackage{tikz}
\usetikzlibrary{positioning}
\usepackage{refstyle}
\usetikzlibrary{arrows.meta}

\providecommand{\corollaryname}{Corollary}
\providecommand{\definitionname}{Definition}
\providecommand{\propositionname}{Proposition}
\providecommand{\remarkname}{Remark}
\providecommand{\theoremname}{Theorem}

\makeatother

\theoremstyle{plain}
\newtheorem{thm}{\protect\theoremname}
\theoremstyle{remark}
\newtheorem{rem}[thm]{\protect\remarkname}
\theoremstyle{definition}
\newtheorem{defn}[thm]{\protect\definitionname}
\theoremstyle{plain}
\newtheorem{prop}[thm]{\protect\propositionname}
\newtheorem{cor}[thm]{\protect\corollaryname}
\newref{thm}{
        name      = \RSthmtxt,
        names     = \RSthmstxt,
        Name      = \RSThmtxt,
        Names     = \RSThmstxt,
        rngtxt    = \RSrngtxt,
        lsttwotxt = \RSlsttwotxt,
        lsttxt    = \RSlsttxt
}

\newref{rem}{
        name      = \RSremtxt,
        names     = \RSremstxt,
        Name      = \RSRemtxt,
        Names     = \RSRemstxt,
        rngtxt    = \RSrngtxt,
        lsttwotxt = \RSlsttwotxt,
        lsttxt    = \RSlsttxt
}

\newref{def}{
        name      = \RSdefntxt,
        names     = \RSdefnstxt,
        Name      = \RSDefntxt,
        Names     = \RSDefnstxt,
        rngtxt    = \RSrngtxt,
        lsttwotxt = \RSlsttwotxt,
        lsttxt    = \RSlsttxt
}

\newref{prop}{
        name      = \RSproptxt,
        names     = \RSpropstxt,
        Name      = \RSProptxt,
        Names     = \RSPropstxt,
        rngtxt    = \RSrngtxt,
        lsttwotxt = \RSlsttwotxt,
        lsttxt    = \RSlsttxt
}

\newref{cor}{
        name      = \RScortxt,
        names     = \RScorstxt,
        Name      = \RSCortxt,
        Names     = \RSCorstxt,
        rngtxt    = \RSrngtxt,
        lsttwotxt = \RSlsttwotxt,
        lsttxt    = \RSlsttxt
}

\def\RScortxt{corollary~}
\def\RScorstxt{corollaries~}
\def\RSCortxt{Corollary~}
\def\RSCorstxt{Corollaries~}
\def\RSdefntxt{definition~}
\def\RSdefnstxt{definitions~}
\def\RSDefntxt{Definition~}
\def\RSDefnstxt{Definitions~}
\def\RSproptxt{proposition~}
\def\RSpropstxt{propositions~}
\def\RSProptxt{Proposition~}
\def\RSPropstxt{Propositions~}
\def\RSremtxt{remark~}
\def\RSremstxt{remarks~}
\def\RSRemtxt{Remark~}
\def\RSRemstxt{Remarks~}
\def\RSthmtxt{theorem~}
\def\RSthmstxt{theorems~}
\def\RSThmtxt{Theorem~}
\def\RSThmstxt{Theorems~}
\providecommand{\corollaryname}{Corollary}
\providecommand{\definitionname}{Definition}
\providecommand{\propositionname}{Proposition}
\providecommand{\remarkname}{Remark}
\providecommand{\theoremname}{Theorem}

\begin{document}
\title{ocLTL: LTL Realizability and Synthesis Modulo $\omega$-Categorical
Structures}
\author{Ohad Asor\thanks{ohad@idni.org}}
\maketitle
\begin{abstract}
We introduce $\mathsf{ocLTL}$, the case of LTL+P modulo $\omega$-categorical
theories. We reduce its realizability and synthesis problems into
the corresponding problems in propositional LTL+P. The core of the
reduction replaces each data subformula with a finite disjunction
over complete types. The complexity remains $\text{2-EXPTIME}$ with
an additional blowup that depends only on the theory but not the formula.
We demonstrate an application of this framework that is related to
atomless Boolean algebras and Lindenbaum-Tarski algebras while drawing
a connection to AI safety.
\end{abstract}

\section{Introduction}

LTL modulo infinite data domain is an active area of research \cite{DD07,BDMSDS11,DL09,FKPS19,RS24,GGG22,ADPS25}.
There have been some decidability and semi-decidability \cite{DD07,BDMSDS11,EF21,KMB18,EFK22,RS24,GGG22}
and undecidability \cite{BDMSDS11,DL09,KK19,BP22} results. In this
paper we focus on realizability and synthesis of LTL modulo $\omega$-categorical
structures. In contrast to some works \cite{RS24,FKPS19,FHP22}, we
do allow direct comparisons of data values across time steps within
data predicates, in a bounded fashion. Full LTL+Past over infinite
traces is supported. The key tool is the Ryll-Nardzewski theorem:
a countable structure is $\omega$-categorical iff it has finitely
many complete $k$-types for each $k$, all of which are isolated.

\section{Preliminaries}

Fix $\mathcal{M}$ to be a countable $\omega$-categorical structure
over a signature $\Sigma$ and domain $M$. We denote by $T_{k}$
the set of all complete $k$-types in $\mathcal{M}$. 
\begin{thm}[Ryll-Nardzewski, Engeler, Svenonius \cite{Hodges93}]
\label{thm:rn}The following are equivalent:
\begin{enumerate}
\item $\mathcal{M}$ is $\omega$-categorical.
\item For each $k$, $\mathcal{M}$ has only finitely many complete $k$-types.
\item Each complete $k$-type in $\mathcal{M}$ contains only finitely many
formulas up to logical equivalence (i.e. it is \emph{isolated}).
\item Two tuples from $\mathcal{M}$ have the same type, iff they lie in
the same $\mathrm{Aut}\left(\mathcal{M}\right)$-orbit.
\end{enumerate}
\end{thm}

\begin{rem}
\label{rem:iota}By \thmref{rn}, any formula $\phi\left(\bar{x}\right)$
with $k$ free variables can be encoded as a finite set of types.
Define $\iota\left(\phi\right)=\left\{ \tau\in T_{k}\mid\phi\in\tau\right\} $.
Then for any $k$-tuple $\bar{a}$, $\mathcal{M}\models\phi\left(\bar{a}\right)$
iff $\mathrm{tp}\left(\bar{a}\right)\in\iota\left(\phi\right)$. Further,
$\phi$ is logically equivalent to $\bigvee_{\tau\in\iota\left(\phi\right)}\tau$.
\end{rem}

\begin{defn}[Restriction]
\label{def:restr}For $\tau\in T_{n}$ and an index set $I\subseteq\left\{ 1,\dots,n\right\} $,
the restriction $\tau|_{I}\in T_{\left|I\right|}$ is the complete
type obtained by taking only formulas in the type that mention only
the variables $\left\{ x_{i}\right\} _{i\in I}$, after appropriate
renumeration of the indices $x_{1},\dots,x_{\left|I\right|}$.
\end{defn}

It is easy to see that the restriction gives a complete type indeed,
given that the original type is complete. 
\begin{prop}[Extension]
\label{prop:ext}For every $\overline{a}\in M^{k}$ and every $\tau\in T_{k+1}$
s.t. $\mathcal{M}\models\tau|_{\left\{ 1,\dots,k\right\} }\left(\bar{a}\right)$
there exists $b\in M$ s.t. $\mathcal{M}\models\tau\left(\bar{a},b\right)$. 
\end{prop}

\begin{proof}
Since $\tau$ is realized, some $\left(\bar{a}',b'\right)$ has type
$\tau$. Then $\mathrm{tp}\left(\bar{a}'\right)=\tau|_{\left\{ 1,\dots,k\right\} }=\mathrm{tp}\left(\bar{a}\right)$,
so by \thmref{rn}, some automorphism $\alpha$ maps $\bar{a}'$ to
$\bar{a}$. Then $\mathrm{tp}\left(\bar{a},\alpha\left(b'\right)\right)=\mathrm{tp}\left(\bar{a}',b'\right)=\tau$. 
\end{proof}

\begin{defn}[Product structure]
\label{def:prod}For $n>0$ define the power structure $\mathcal{M}^{n}$
to have the set of all $n$-tuples from $M$, inheriting the same
signature extended with the projection functions. 
\end{defn}

The following is immediate too since $\Sigma$ is still operating
over single elements via the projection functions: 
\begin{prop}
\label{prop:prod}If $\mathcal{M}$ is $\omega$-categorical then
so is $\mathcal{M}^{n}$. Moreover, there is a canonical bijection
between the $k$-types of $\mathcal{M}^{n}$ and the $kn$-types of
$\mathcal{M}$.
\end{prop}

This extends to a product of distinct $\omega$-categorical structures
too.

Finally, we assume $\mathcal{M}$ is \emph{effectively presented}:
the complete type of any tuple is computable, and the witness $b$
in proposition \propref{ext} can be computed from $\bar{a}$ and
$\tau$.

\section{$\mathsf{ocLTL}\left(\mathcal{M}\right)$}

A reactive system reads inputs and produces outputs, at each time
step. Inputs and outputs take values in a data domain $M$. For simplicity
we assume the number of input and output streams are equal, denote
it by $s$. A \emph{program} is a function $F:\left(M^{\omega}\right)^{s}\rightarrow\left(M^{\omega}\right)^{s}$
that is \emph{causal }, i.e.
\[
\forall n\forall a,b\in\left(M^{\omega}\right)^{s}.\text{pr}_{n}\left(a\right)=\text{pr}_{n}\left(b\right)\to\text{pr}_{n}\left(F\left(a\right)\right)=\text{pr}_{n}\left(F\left(b\right)\right)
\]
where $\text{pr}_{n}$ is the prefix function returning the first
$n$ elements of a word, and its generalization to tuples of words
extends componentwise.

\subsection{Syntax and Semantics}

Fix a \emph{lookback} parameter $\ell\geq0$. A \emph{data formula}
$\delta$ is a first-order $\Sigma$-formula whose free variables
are $x^{i}_{-j},y^{i}_{-j}$ for $1\leq i\leq s$ and $0\leq j\leq\ell$.
The $x$'s are the \emph{inputs} and the $y$'s are the \emph{outputs}.
A trace is a function $\pi:\mathbb{N}\to M^{s}\times M^{s}$. Denote
$\pi_{x^{i}}\left(t\right)$ and $\pi_{y^{i}}\left(t\right)$ for
the $i$'th input and output at time $t$, respectively. 
\begin{rem}
It is easy to support lookahead as well by allowing $x^{i}_{+j},y^{i}_{+j}$
as well, but we keep the lookback-only setting for simplicity.
\end{rem}

\begin{defn}
\label{def:ltlm}$\mathsf{ocLTL}\left(\mathcal{M}\right)$ formulas
have the syntax
\[
\varphi\coloneqq\delta\mid\neg\varphi\mid\varphi_{1}\wedge\varphi_{2}\mid\mathsf{X}\varphi\mid\varphi_{1}\mathsf{U}\varphi_{2}\mid\mathsf{Y}\varphi\mid\varphi_{1}\mathsf{S}\varphi_{2}
\]

under the semantics
\begin{align*}
\pi,t\ge\ell\models\delta\iff & \mathcal{M}\models\delta\left[x^{i}_{-j}\mapsto\pi_{x^{i}}\left(t-j\right),y^{i}_{-j}\mapsto\pi_{y^{i}}\left(t-j\right)\right]\\
\pi,t<\ell\models\delta\iff & \bot\\
\pi,t\models\neg\varphi\iff & \pi,t\not\models\varphi\\
\pi,t\models\varphi_{1}\wedge\varphi_{2}\iff & \pi,t\models\varphi_{1}\wedge\pi,t\models\varphi_{2}\\
\pi,t\models\mathsf{X}\varphi\iff & \pi,t+1\models\varphi\\
\pi,t\models\varphi_{1}\mathsf{U}\varphi_{2}\iff & \exists k\geq t.\pi,k\models\varphi_{2}\wedge\forall t\leq j<k.\pi,j\models\varphi_{1}\\
\pi,t\models\mathsf{Y}\varphi\iff & t>0\wedge\pi,t-1\models\varphi\\
\pi,t\models\varphi_{1}\mathsf{S}\varphi_{2}\iff & \exists k\leq t.\pi,k\models\varphi_{2}\wedge\forall k<j\leq t.\pi,j\models\varphi_{1}
\end{align*}

We write $\mathsf{G}\varphi\coloneqq\neg\left(\top\mathsf{U}\neg\varphi\right)$
and $\mathsf{F}\varphi\coloneqq\top\mathsf{U}\varphi$. A program
\emph{realizes }$\varphi$ if every induced trace $\pi$ satisfies
$\pi,0\models\varphi$.
\end{defn}

This definition is similar to the one in \cite{GGG22} except three
differences: we describe the time-shifts using a different notation
rather their $\ocircle$, we work with $\omega$-words rather finite
traces, and we extend the definition beyond LTL into LTL+P.
\begin{rem}
Speaking about synthesis alongside model-theoretic types introduces
an unfortunate terminological clash. We can speak about a program
realizing a specification. But when we say that types are realized
in a model, or realized by a tuple of domain elements, we refer to
a different concept. The meaning of ''realize'' will always be clear
from the context. 
\end{rem}

\subsection{Three Variables}\label{subsec:3var}

To avoid cumbersome notation we assume that data subformulas contain
each only three free variables: $m,x,y$ (memory, current input, current
output). Bringing formulas to this form straightforward: first introduce
auxiliary output streams, each of them responsible to carry the values
from the bounded past (this can easily be expressed in $\mathsf{ocLTL}\left(\mathcal{M}\right)$).
Then pack everything into a product structure (to collapse multiple
input and output streams, including the auxiliary ones, into single
input and output streams) and obtain the desired form. It is easy
to generalize all derivations in this paper to the general case, one
that involves neither the reduction into three variables nor moving
to the product structure, but the notation becomes more complex.

\section{Realizability}\label{sec:Realizability}

In this section we shall reduce $\mathsf{ocLTL}\left(\mathcal{M}\right)$
into an equirealizable $\mathsf{LTL}$ formula. As concluded in \subsecref{3var},
we assume that all data subformulas have $m,x,y$ only as free variables,
so in particular $s=\ell=1$. By \thmref{rn}, we can now write each
data subformula as a disjunction of types from $T_{3}$ (represented
by the propositional variables $Q$'s below). After doing so, we replace
each type with a propositional variable. The inputs in the translated
$\mathsf{LTL}$ formula will encode types from $T_{2}$ (represented
by the propositional variables $P$'s below) representing the current
\emph{partial state} $\left(m,x\right)$ which is the previous output
(the memory) and the current input. This gives rise to an assume-guarantee
structure: we assume that the inputs are well-formed in this fashion,
e.g. by an encoder sitting between the program and the environment.
The outputs will encode the \emph{full state} $\left(m,x,y\right)$
which appends the current output $y$ to the partial state. We will
then encode in $\mathsf{LTL}$ how those variables interact: each
tuple should realize exactly one type, the partial state type should
conform to the type before it (the assume-guarantee structure), and
the full state type should conform to the partial state. This all
is formalized as follows: 
\begin{enumerate}
\item \label{enu:sub}Replace each data subformula $\delta$ with $\bigvee_{\tau\in\iota\left(\delta\right)}Q_{\tau}$
and obtain $\varphi^{*}$ (cf. \remref{iota}).
\item Express the uniqueness constraints: 
\begin{align*}
\Phi_{I}\coloneqq & \bigvee_{\sigma\in T_{2}}P_{\sigma}\wedge\bigwedge_{\sigma_{1}\neq\sigma_{2}}\neg\left(P_{\sigma_{1}}\wedge P_{\sigma_{2}}\right)\\
\Phi_{O}\coloneqq & \bigvee_{\tau\in T_{3}}Q_{\tau}\wedge\bigwedge_{\tau_{1}\neq\tau_{2}}\neg\left(Q_{\tau_{1}}\wedge Q_{\tau_{2}}\right)
\end{align*}
since each tuple realizes exactly one type. 
\item Express that the type of the full state $\left(m,x,y\right)$ and
the type of the partial state $\left(m,x\right)$ agree on $m,x$
(cf. \figref{fig1}): 
\[
\Psi_{O}\coloneqq\bigwedge_{\tau\in T_{3}}Q_{\tau}\to P_{\tau|_{m,x}}
\]
\item Similarly, express their agreement on the current $y$ which becomes
the next-state $m$:
\[
\Psi_{I}\coloneqq\bigwedge_{\tau\in T_{3}}Q_{\tau}\to\mathsf{X}\bigvee_{\sigma\in S_{\tau}}P_{\sigma}
\]
where 
\[
S_{\tau}=\left\{ \sigma:\sigma|_{m}=\mathsf{sh}\left(\tau\right)|_{m}\right\} 
\]
and $\mathsf{sh}\left(\tau\right)$ takes $\tau\left(m,x,y\right)$
and ''shifts'' it by renaming its variables into $\tau\left(u,v,m\right)$.
The variables $u,v$ are unused because we immediately restrict to
$m$ (cf. \figref{fig1}).
\item The translated specification is: 
\[
\hat{\varphi}\coloneqq\mathsf{G}\left(\Psi_{I}\wedge\Phi_{I}\right)\to\left(\varphi^{*}\wedge\mathsf{G}\left(\Psi_{O}\wedge\Phi_{O}\right)\right).
\]
\end{enumerate}
Declare the $P$'s as inputs and the $Q$'s as outputs. The result
has $\left|T_{2}\right|+\left|T_{3}\right|$ propositional variables.

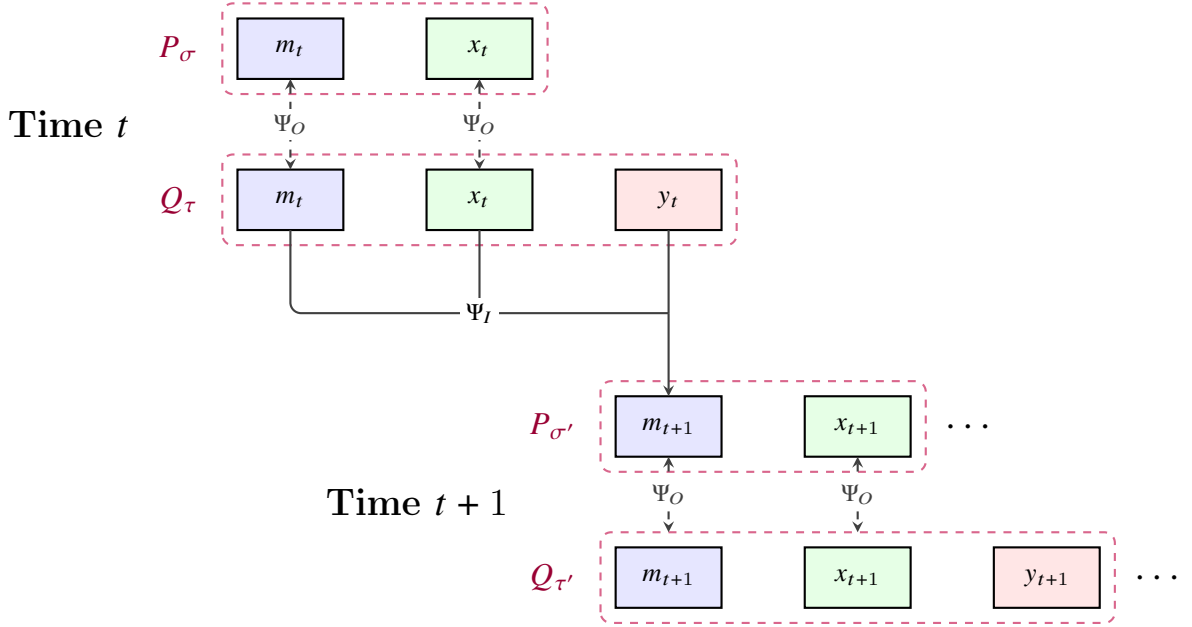
\begin{figure}[tbh]
\centering \begin{tikzpicture}[>=stealth]

    \tikzset{
        varblock/.style={draw, thick, minimum height=0.8cm, minimum width=1.4cm, text centered, font={\sffamily}},
        mblock/.style={varblock, fill=blue!10},
        xblock/.style={varblock, fill=green!10},
        yblock/.style={varblock, fill=red!10},
        domainbox/.style={draw=purple!60, dashed, thick, rounded corners}
    }

    \node[font={\bfseries\Large}, anchor=east] at (-2.0, 5.0) {Time $t$};
    \node[font={\bfseries\Large}, anchor=east] at (3.0, 0.0) {Time $t+1$};

    \node[mblock] (m_env_t) at (0, 6.0) {$m_t$};
    \node[xblock] (x_env_t) at (2.5, 6.0) {$x_t$};
    
    \draw[domainbox] (-0.9, 5.4) rectangle (3.4, 6.6);
    \node[font={\large\bfseries}, text=purple!80!black, anchor=east] at (-1.1, 6.0) {$P_{\sigma}$};

    \node[mblock] (m_sys_t) at (0, 4.0) {$m_t$};
    \node[xblock] (x_sys_t) at (2.5, 4.0) {$x_t$};
    \node[yblock] (y_sys_t) at (5.0, 4.0) {$y_t$};
   
    \draw[domainbox] (-0.9, 3.4) rectangle (5.9, 4.6);
    \node[font={\large\bfseries}, text=purple!80!black, anchor=east] at (-1.1, 4.0) {$Q_{\tau}$};

    \draw[<->, thick, dashed, darkgray] (0, 5.6) -- node[fill=white, font={\small\bfseries}, inner sep=2pt] {$\Psi_O$} (0, 4.4);
    \draw[<->, thick, dashed, darkgray] (2.5, 5.6) -- node[fill=white, font={\small\bfseries}, inner sep=2pt] {$\Psi_O$} (2.5, 4.4);

    \node[mblock] (m_env_next) at (5.0, 1.0) {$m_{t+1}$};
    \node[xblock] (x_env_next) at (7.5, 1.0) {$x_{t+1}$};
    
    \draw[domainbox] (4.1, 0.4) rectangle (8.4, 1.6);
    \node[font={\large\bfseries}, text=purple!80!black, anchor=east] at (3.9, 1.0) {$P_{\sigma'}$};

    \node[mblock] (m_sys_next) at (5.0, -1.0) {$m_{t+1}$};
    \node[xblock] (x_sys_next) at (7.5, -1.0) {$x_{t+1}$};
    \node[yblock] (y_sys_next) at (10.0, -1.0) {$y_{t+1}$};
    
    \draw[domainbox] (4.1, -1.6) rectangle (10.9, -0.4);
    \node[font={\large\bfseries}, text=purple!80!black, anchor=east] at (3.9, -1.0) {$Q_{\tau'}$};

    \draw[<->, thick, dashed, darkgray] (5.0, 0.6) -- node[fill=white, font={\small\bfseries}, inner sep=2pt] {$\Psi_O$} (5.0, -0.4);
    \draw[<->, thick, dashed, darkgray] (7.5, 0.6) -- node[fill=white, font={\small\bfseries}, inner sep=2pt] {$\Psi_O$} (7.5, -0.4);

    \draw[thick, darkgray, rounded corners] (0, 3.6) -- (0, 2.5) -- (5.0, 2.5);
    \draw[thick, darkgray] (2.5, 3.6) -- (2.5, 2.5);
    \draw[thick, darkgray] (5.0, 3.6) -- (5.0, 2.5);
    
    \draw[->, thick, darkgray] (5.0, 2.5) -- (5.0, 1.4);
    
    \node[fill=white, font={\small\bfseries}, inner sep=2pt] at (2.5, 2.5) {$\Psi_I$};
    
    \node[font={\Large}] at (9.0, 1.0) {$\dots$};
    \node[font={\Large}] at (11.5, -1.0) {$\dots$};
\end{tikzpicture} \caption{\textbf{Type-compatibility constraints 3,4}. At time $t$, the environment
declares a $2$-type $P_{\sigma}=\mathrm{tp}\left(m_{t},x_{t}\right)$
and the system responds with a $3$-type $Q_{\tau}=\mathrm{tp}\left(m_{t},x_{t},y_{t}\right)$.
The constraint $\Psi_{O}$ requires that $\tau$ restricts to $\sigma$
on $\left(m,x\right)$: the system's output type must extend the environment's.
The constraint $\Psi_{I}$ links consecutive steps: since $y_{t}$
becomes the memory $m_{t+1}$, the $1$-type $\mathsf{sh}\left(\tau\right)|_{m}$
determined by the current output must equal the memory component $\sigma'|_{m}$
of any input type $P_{\sigma'}$ the environment inputs at time $t+1$.}
\label{fig:fig1} 
\end{figure}

\begin{thm}[Equirealizability]
\label{thm:main}$\varphi$ is realizable over $\mathcal{M}$ iff
$\hat{\varphi}$ is realizable. 
\end{thm}

\begin{proof}
Immediate from the construction: each direction uses proposition \propref{ext}
and the effective presentation assumption to find concrete witnesses
matching a prescribed type, and \thmref{rn} to transfer truth values
between the concrete and propositional traces. We nevertheless present
a formal proof:

Assume $F$ realizes $\varphi$. We construct a propositional strategy
$G$ realizing $\hat{\varphi}$. $G$ only needs to satisfy $\Psi_{O}\wedge\Phi_{O}\wedge\varphi^{*}$
whenever the environment satisfies the $\Psi_{I}\wedge\Phi_{I}$.
At $t=0$ let $m_{0}\in M$ be arbitrary. For $t>0$ the environment
provides a valuation for the $P$ variables. By $\Phi_{I}$, exactly
one $P_{\sigma}$ holds. $\Psi_{I}$ guarantees that $\sigma_{t}|_{m}=\mathsf{sh}\left(\tau_{t-1}\right)|_{m}$.
Since $\tau_{t-1}$ was realized by $m_{t-1},x_{t-1},y_{t-1}$, setting
$m_{t}=y_{t-1}$ realizes the $m$-component of $\sigma_{t}$. By
\propref{ext}, there exists a concrete $x_{t}\in M$ such that $\mathrm{tp}\left(m_{t},x_{t}\right)=\sigma_{t}$.
We feed the history of inputs up to $x_{t}$ into $F$ which causally
produces an output $y_{t}$. We then compute the complete type $\tau_{t}=\mathrm{tp}\left(m_{t},x_{t},y_{t}\right)$.
$G$ then outputs $Q_{\tau}=\top$ and $Q_{\tau'}=\bot$ for all $\tau'\neq\tau_{t}$,
so $\Phi_{O}$ holds. $\Psi_{O}$ holds because $\tau_{t}|_{m,x}=\mathrm{tp}\left(m_{t},x_{t}\right)=\sigma_{t}$.
Since the concrete trace generated by $F$ satisfies $\varphi$, and
$\varphi^{*}$ is derived by replacing every data formula $\delta$
with $\bigvee_{\tau\in\iota\left(\delta\right)}Q_{\tau}$, \thmref{rn}
guarantees that the sequence of types $\tau_{t}$ satisfies $\varphi^{*}$.

For the other direction, assume $G$ realizes $\hat{\varphi}$. We
construct a program $F$ over $\mathcal{M}$ realizing $\varphi$.
At $t=0$, we have an input $x_{0}$ and we pick an arbitrary $m_{0}\in M$.
For $t>0$, given the memory $m_{t}$ and the new input $x_{t}$,
we compute $\sigma_{t}=\mathrm{tp}\left(m_{t},x_{t}\right)$. We feed
the environment valuation $P_{\sigma_{t}}=\top$ (and all others $\bot$)
to $G$. $\Phi_{I}$ is satisfied by construction. We set $m_{t}=y_{t-1}$
as expected and by that we also ensure that memory part of $\sigma_{t}$
matches $\mathsf{sh}\left(\tau_{t-1}\right)|_{m}$, satisfying $\Psi_{I}$.
By assumption, $G$ now outputs a valuation satisfying $\Psi_{O}\wedge\Phi_{O}\wedge\varphi^{*}$.
By $\Phi_{O}$, exactly one variable $Q_{\tau_{t}}$ is true. By $\Psi_{O}$,
$\tau_{t}|_{m,x}=\sigma_{t}$, so $\tau_{t}$ restricts to $\mathrm{tp}\left(m_{t},x_{t}\right)$.
\propref{ext} guarantees a $y_{t}\in M$ such that $\mathrm{tp}\left(m_{t},x_{t},y_{t}\right)=\tau_{t}$.
Using the effective presentation assumption, $F$ computes such a
$y_{t}$ and outputs it, setting $m_{t+1}=y_{t}$ for the next step.
This creates a concrete trace over $\mathcal{M}$ where each tuple
$\left(m_{t},x_{t},y_{t}\right)$ realizes the type $\tau_{t}$ chosen
by $G$. Since $G$ ensures that the sequence $Q_{\tau}$ satisfies
$\varphi^{*}$, \thmref{rn} implies that the concrete trace satisfies
the original data formulas, so $F$ realizes $\varphi$. 
\end{proof}

$\mathcal{M}$ is fixed, hence the blowup is a function of the number
of types and is therefore a constant, but this is only after the reduction
into the $3$-variable form, so $\mathcal{M}$ is a product structure
whose dimension determined by $s,\ell$. We obtain: 
\begin{cor}
$\mathrm{LTL}\left(\mathcal{M}\right)$ realizability is $\text{2-EXPTIME}$
complete for fixed $\mathcal{M},s,\ell$. Similarly satisfiability
is complete for PSPACE.
\end{cor}

Synthesis is now straightforward by construction: given a program
that realizes the propositional LTL formula, we reproduce the input
encoding and the output decoding according to the types mapped to
the propositional variables.

\section{Optimizations}

\subsection{Binary Encoding}

One trivial optimization is to encode each $2$-type using $\left\lceil \log_{2}\left|T_{2}\right|\right\rceil $
propositional variables (the type ''number'' encoded in binary)
and similarly for $3$-types. Instead of $\left|T_{2}\right|+\left|T_{3}\right|$
propositional variables we obtain only $\left\lceil \log_{2}\left|T_{2}\right|\right\rceil +\left\lceil \log_{2}\left|T_{3}\right|\right\rceil $
many, and further we do not need to encode that at exactly one type
holds at each point in time as this is automatic from this encoding
for both existence and uniqueness.

\subsection{Avoiding $T_{3}$}

Avoiding enumerating $T_{3}$ in the worst case is not possible: we
have formulas with $3$ free variables and we have to keep track of
which hold or do not hold at each step. In the worst case, the Boolean
algebra generated by those formulas is precisely the one generated
by $T_{3}$. However, if the generated algebra is coarser, we can
take advantage of that.

In what follows let $\Delta$ be the set of data subformulas in $\varphi.$
For $A\in\left\{ \bot,\top\right\} ^{\left|\Delta\right|}$, define
the $\Delta$-minterm $\Delta^{A}$ by 
\[
\Delta^{A}\coloneqq\bigwedge_{\delta|a_{\delta}=\top}\delta_{i}\left(m,x,y\right)\wedge\bigwedge_{\delta|a_{\delta}=\bot}\neg\delta_{i}\left(m,x,y\right)
\]
Observe the structure of $\hat{\varphi}$ from \secref{Realizability}:
the output propositions $Q_{\tau}$ appear in three places. In $\varphi^{*}$,
each data subformula $\delta$ was replaced by $\bigvee_{\tau\in\iota\left(\delta\right)}Q_{\tau}$,
which simply encodes whether $\delta$ holds. In $\Psi_{I}$, the
output type determines the successor memory type $\mathsf{sh}\left(\tau\right)|_{m}\in T_{1}$,
which constrains which input types are admissible at the next step.
In $\Psi_{O}$ and $\Phi_{O}$, the output type must be unique and
extend the input type.

So $\hat{\varphi}$ does not need the full type $\tau$. It needs
only two things: which of the $\left|\Delta\right|$ data subformulas
hold, and what the successor memory type is. Replace the $\left|T_{3}\right|$
output propositions $Q_{\tau}$ with $\left|\Delta\right|+\left|T_{1}\right|$
propositions: $D_{i}$ for each $\delta_{i}\in\Delta$, and $R_{\rho}$
for each $\rho\in T_{1}$ encoding the successor's memory type. In
$\varphi^{*}$, replace $\bigvee_{\tau\in\iota\left(\delta_{i}\right)}Q_{\tau}$
with simply $D_{i}$. In $\Psi_{I}$, replace $\mathsf{sh}\left(\tau\right)|_{m}$
with $\rho$. Replace $\Phi_{O}$ with uniqueness for $R_{\rho}$:
\[
\Phi^{R}_{O}\coloneqq\bigvee_{\rho\in T_{1}}R_{\rho}\wedge\bigwedge_{\rho_{1}\neq\rho_{2}}\neg\left(R_{\rho_{1}}\wedge R_{\rho_{2}}\right)
\]
or skip it if using a binary encoding.

So far, not every combination $\left(\rho,\sigma,D_{1},\dots,D_{\left|\Delta\right|}\right)$
is witnessed by an actual type. The output could declare $D_{1}\wedge D_{2}$
even when $\delta_{1}\wedge\delta_{2}$ is unsatisfiable. We therefore
add the constraint: 
\[
\Phi_{\delta}\coloneqq\bigwedge_{\sigma,\rho,A\in\left\{ \bot,\top\right\} ^{\left|\Delta\right|}}\left(P_{\sigma}\wedge R_{\rho}\wedge{\cal D}^{A}\right)\to\underbrace{\exists mxy.\sigma\left(m,x\right)\wedge\rho\left(y\right)\wedge\Delta^{A}}_{\text{precomputed}}
\]
where ${\cal D}^{A}$ is minterm in the $D$'s in the same fashion
as $\Delta^{A}$. Each infeasible combination contributes one forbidden-pattern
clause. The guard $\Phi_{\delta}$ subsumes $\Psi_{O}$: any feasible
$\left(\sigma,\rho,a_{1},\dots,a_{K}\right)$ is witnessed by a $\tau$
with $\tau|_{m,x}=\sigma$, so the requirement that the output type
extends the input type is already enforced by feasibility. The translated
specification now becomes 
\[
\mathsf{G}\left(\Psi_{I}\wedge\Phi_{I}\right)\to\left(\varphi^{*}\wedge\mathsf{G}\left(\Phi^{R}_{O}\wedge\Phi_{\delta}\right)\right).
\]

The total number of propositional variables is now $\left|T_{1}\right|+\left|T_{2}\right|+\left|\Delta\right|$.
Combined with the binary encoding optimization, we can reduce it into
$\left\lceil \log_{2}\left|T_{1}\right|\right\rceil +\left\lceil \log_{2}\left|T_{2}\right|\right\rceil +\left|\Delta\right|$.

\section{Extensions to the Data Language}

We show two ways in which we can extend the expressiveness of given
$\omega$-categorical theories: constant symbols, and fixpoint operators.

The first case is quite trivial: adding any finite number $c$ of
interpreted and uninterpreted constants to the language maintains
$\omega$-categoricity, since the $k$-types in the expanded structure
correspond to the $\left(k+c\right)$-types in the original structure,
and the claim follows immediately from \thmref{rn}. For a similar
result cf. \cite{Hodges93} Theorem 7.3.1, p. 341.

Next we that the language can be extended with fixed-point operators.
Unlike adding constants, this does not add expressivity but adds succinctness.
PFP (partial fixed-point) operators were studied in \cite{AV89} for
finite structures and were enhanced to infinite structures in \cite{Kreutzer02}.
For the sake of generality we argue for PFP and the results for LFP
and GFP follow. Denote by $\mathcal{L}\left(\mathcal{M}\right)$ the
first-order language of $\mathcal{M}$.
\begin{defn}
\label{def:pfp}$\mathcal{L}_{\text{PFP}}\left(\mathcal{M}\right)$
is the language $\mathcal{L}\left(\mathcal{M}\right)$ extended with
the \emph{partial fixed-point operator}. Specifically, if $\phi$
is a formula with a free relation symbol $R$ and free variable symbols
$X$, and $Y$ is another tuple of variable symbols, then $\left(\text{pfp}_{R,X}\phi\right)\left(Y\right)$
is a formula. Given $\mathcal{U}$ an interpretation of $U$, define
$R^{0}\coloneqq\emptyset$ and $R^{n+1}\coloneqq\left\{ Y|\mathcal{M},\mathcal{U}\models\phi\left(R^{n},Y\right)\right\} $.
Then $\mathcal{M},\mathcal{U}\models\left(\text{pfp}_{R,X}\phi\right)\left(Y\right)$
iff the sequence $R^{n}$ has a fixed point and $Y$ is a tuple in
that fixed point.
\end{defn}

\begin{thm}
Over $\omega$-categorical structures, every formula in $\mathcal{L}_{\text{PFP}}\left(\mathcal{M}\right)$
can effectively be written as a logically equivalent formula in $\mathcal{L}\left(\mathcal{M}\right)$.
\end{thm}

\begin{proof}
Seeing each $R^{n}$ as a relation, and observing that $R^{n+1}\coloneqq\left\{ Y|\mathcal{M},\mathcal{U}\models\phi\left(R^{n},Y\right)\right\} $
is the same as saying $\mathcal{M},R^{n},R^{n+1}\models R^{n+1}\left(Y\right)\leftrightarrow\phi\left(R^{n},Y\right)$,
set $\psi_{0}\coloneqq\bot$ and 
\[
\psi_{n+1}\left(X\right)\coloneqq\phi\left(\psi_{n}\left(X\right),X\right)
\]
Observing that all $\psi_{n}\left(X\right)$ belong to $\mathcal{L}\left(\Sigma,M_{1},U_{1}\right)$
where $M_{1}$ ($U_{1}$) are all {[}finitely many{]} interpreted
(uninterpreted) constants appearing in $\phi$, then by \thmref{rn}
exists $N$ s.t. $\forall n>N\exists k\leq N.\psi_{n}\equiv\psi_{k}$.
For $m=2,\dots,N$ define
\[
\chi_{m}\left(X\right)\coloneqq\psi_{m}\left(X\right)\wedge\forall Y.\psi_{m}\left(Y\right)\leftrightarrow\psi_{m-1}\left(Y\right)
\]
and now $\chi_{N}\equiv\text{pfp}_{R,X}\phi$ by construction.
\end{proof}

\begin{cor}
$\mathcal{L}\left(\mathcal{M}\right)$ can be extended with LFP and
GFP operators which in turn can be eliminated.
\end{cor}

\begin{proof}
The LFP case is a special case of PFP. For GFP we use the same proof
under the modification $\psi_{0}\coloneqq\top$. For both LFP and
GFP, the construction can be more efficient by considering $\psi_{N}$
without the need to construct $\chi$, since a fixed point is guaranteed
to exist.
\end{proof}

\begin{rem}
Observe that, in contrast to many LFP/GFP logics, we do not require
$\phi$ to be syntactically monotone, e.g. that $R$ appears only
under an even number of negations. But in our setting semantic monotonicity
is decidable and is trivially expressible in the language itself.
\end{rem}

\section{Application}

We sketch an application related to Lindenbaum-Tarski algebra (LTA)
of logics in general. One major example of an $\omega$-categorical
structure is the one of atomless Boolean algebra (BA). The LTA of
virtually any classical logic $\mathcal{L}$ with infinite signature
is atomless BA. Therefore our setting allows to specify $\text{\ensuremath{\mathsf{ocLTL\left[LTA\left[\mathcal{L}\right]\right]}}}$programs
whose input and output objects are formulas, allowing Boolean operations
between them, as well as equality understood as logical equivalence.
Moreover, using the product structure construction, several logics
can be supported at once alongside additional $\omega$-categorical
structures.

This gives rise to semantic homoiconicity: inputs and outputs can
be formulas in this very language itself, so the language is $\text{\ensuremath{\mathsf{ocLTL\left[LTA\left[ocLTL\left[LTA\left[\dots\right]\right]\right]\right]}}}$,
bypassing paradoxes arising from self-reference, because the setting
is limited to the language of BA only. Strictly speaking, making the
signature infinite e.g. by adding infinitely many uninterpreted constants,
breaks $\omega$-categoricity. However this can be mitigated by recalling
that each formula is finite and then using the fact that finitely
many constants maintain $\omega$-categoricity as we showed in the
previous section, and using Boolean-algebraic considerations that
we do not mention here.

This can be taken a step forward towards a system that takes as inputs
new constraints about itself, judging whether or not to accept or
reject them using Boolean operations against other sentences in the
language (e.g. checking entailment against security conditions), and
whenever accepted, applies the update to itself. For example, the
system accepts commands from the user, each command is an LTL spec
itself, while the system is also implemented using LTL specs, those
commands may refer to other commands, e.g. a command can say ``if
any of my future commands contradicts this given condition, then do
not execute it''. This is a part of a blueprint towards Safe AI.
We do not elaborate on this setting here, but we have already implemented
such a system\footnote{http://github.com/idni/tau-lang} (at the time
of writing, the full LTL support in this system is still pending and
will be shipped over the next weeks).

\section{Related Work}

LTL modulo theories is an active area of research. We mention only
some of the works closest to our approach: 
\begin{enumerate}
\item RealMT: \cite{RS24} reduces LTL synthesis modulo $\exists^{*}\forall^{*}$-decidable
theories to propositional synthesis via propositional abstraction.
However, RealMT variables are \emph{fresh} at each step: the formula
$\mathsf{G}\left(y_{0}>y_{-1}\right)$ is inexpressible because it
compares outputs across time steps. For lookback $\ell=0$, RealMT
subsumes our framework. For $\ell>0$ our framework seems to be the
only result so far providing decidable synthesis with cross-step data
references over infinite domains. 
\item LTL modulo theories: \cite{GGG22} studies $\mathrm{LTL}_{f}$ modulo
theories over \emph{finite} traces, giving only semi-decidability
results for satisfiability when data constraints come from decidable
SMT theories. \cite{ADPS25} extends this to full LTL synthesis over
infinite-state arenas using a CEGAR loop: abstract via Boolean predicates,
synthesize, check spuriousness, and refine. \cite{CGGMT24} further
presents a bounded fragment admitting 2EXPTIME synthesis (EXPTIME
without constants).
\item Register automata and data words: synthesis over infinite alphabets
has been studied through register automata. \cite{EF21} considers
synthesis of register transducers over linearly ordered data domains
(e.g. discrete and dense orders), using a finite number of registers
to store data values. \cite{KK19} studies synthesis for register
automata specifications with universality constraints. They prove
both decidability and undecidability results, depending on register
bounds. In our framework, $r$ registers over $\left(\mathbb{Q},<\right)$
(which is an $\omega$-categorical structure) correspond to lookback
$\ell=r$ with $s=1$ stream. The type spaces grow with $\ell$ but
remain finite for each fixed $\ell$. The register automata approach
does not directly generalize to richer structures while our framework
handles any $\omega$-categorical $\mathcal{M}$. 
\item Data-word logics: Freeze LTL \cite{DL09} extends LTL with a freeze
quantifier that stores the current data value in a register. Over
infinite data words, satisfiability is already undecidable with a
single register and only the $\mathsf{X}$ and $\mathsf{F}$ operators.
Our bounded lookback $\ell$ avoids this: the number of simultaneously
accessible past values is fixed, preventing the encoding of counter
machines. This is a deliberate design choice that enables decidability. 
\item Orbit-Finiteness: \cite{BKL14} develops an automata theory parameterized
by a symmetry group $G$ acting on an infinite alphabet. An automaton
is \emph{orbit-finite} if its state set has finitely many $G$-orbits.
When $G=\mathrm{Aut}\left(\mathcal{M}\right)$ for an $\omega$-categorical
$\mathcal{M}$, their orbit-finite sets are exactly our type spaces
$T_{k}=M^{k}/\mathrm{Aut}\left(\mathcal{M}\right)$ - not an analogy
but an identity via the Ryll-Nardzewski theorem, even though \cite{BKL14}
does not mention $\omega$-categoricity. For the equality symmetry,
their orbit-finite nondeterministic automata coincide (Theorem 6.4
of \cite{BKL14}) with the finite-memory automata of \cite{KF94},
for which emptiness is decidable but universality is undecidable \cite{KF94}.
The main difference between their and our approaches is that they
study language-theoretic properties (emptiness, inclusion, determinization)
while we study synthesis. Our equirealizability theorem is specific
to the synthesis setting and seems to have no direct counterpart in
the orbit-finite automata literature. 
\end{enumerate}

\section{Conclusion}\label{sec:Conclusion}

We have shown that $\omega$-categoricity gives decidable LTL+P realizability
and synthesis, while allowing data variables to interact over bounded
lookback, and while showing that from a theoretical point of view,
when compared to propositional LTL, only the lookback increases the
complexity. The types are the only $\mathcal{M}$-dependent data.
Once computed, the main realizability and synthesis algorithms are
purely propositional, with type enumeration being the main blowup.

The main beneficiaries of $\omega$-categoricity are fragments of
LTL for which realizability and synthesis can be translated directly
into fixpoint operators, in which case no type enumeration nor $\Delta$-minterms
are needed. A direct corollary of \thmref{rn} is that in $\omega$-categorical
structures, fixpoint operators can be eliminated. Therefore in such
fragments of LTL, there is no need to use propositional encoding,
and the problem can be decided directly using fixed points.

We demonstrated an application related to atomless BA and Lindenbaum-Tarski
algebras, and sketched how similar developments are related to AI
safety.

\section*{Acknowledgments}

I would like to thank Enrico Franconi, Lucca Tiemens, and Paweł Parys,
for plenty of useful discussions.

\section*{Intellectual Property}

Some of the methods discussed here are patented.

\bibliographystyle{plain}
\bibliography{references}

\end{document}